\documentclass[letterpaper, abstracton, DIV11, 12pt, titlepage=false]{scrartcl}

\usepackage[ruled]{algorithm2e}

\SetAlFnt{\small}
\SetAlCapFnt{\small}
\SetAlCapNameFnt{\small}
\SetAlCapHSkip{0pt}
\IncMargin{-\parindent}
\let\chapter\undefined

\usepackage[usenames,dvipsnames]{xcolor}
\usepackage{footnote}
\usepackage{enumerate}
\usepackage[round]{natbib}
\usepackage[english]{babel}
\usepackage{amsfonts}
\usepackage{amsmath}
\usepackage{amssymb}
\usepackage{booktabs}
\usepackage{cancel}
\usepackage{amsthm}
\usepackage{multirow}
\usepackage{mathabx}
\usepackage{dsfont}
\usepackage{graphicx}
\usepackage{tikz}
\usepackage{pifont}
\usepackage{bm}
\usepackage{wrapfig}
\usepackage{microtype}
\usepackage{setspace}
\usepackage[colorlinks=true,bookmarks=true, pdftex, citecolor = blue, linkcolor = blue,urlcolor = blue]{hyperref}

\onehalfspace

\DisableLigatures[f]{encoding = *, family = * }

\def\bf{\normalfont\bfseries}

\changenotsign

\definecolor{darkgreen}{RGB}{40,150,70}

\newcommand{\URBIr}{\text{URBI($r$)}}

\newtheorem{claim}{Claim}

\theoremstyle{plain}
\newtheorem{theorem}{Theorem}

\newtheorem{corollary}{Corollary}

\newtheorem{fact}{Fact}

\theoremstyle{definition}
\newtheorem{definition}{Definition}

\newcommand{\ourrep}{}

\theoremstyle{remark}
\newtheorem{remark}{Remark}

\begin{document}

{\setstretch{1}
\title{%
\LARGE{%
Local Sufficiency for Partial Strategyproofness}\thanks{%
\scriptsize{%
Department of Informatics, University of Zurich, Switzerland,
email: \{mennle, seuken\}@ifi.uzh.ch.
For updates see www.ifi.uzh.ch/ce/publications/LSPSP.pdf. 
We would like to thank 
Utku \"{U}nver for helpful comments on this work.
Any errors remain our own.
Part of this research was supported by the Hasler Foundation under grant \#12078 and the SNSF (Swiss National Science Foundation) under grant \#156836.}}}
\author{%
Timo Mennle \\ University of Zurich
\and Sven Seuken \\ University of Zurich }
\date{First version: October 3, 2015 \\
This version: \today}
\maketitle


\begin{abstract}
In \citep{MennleSeuken2017PSP_WP}, we have introduced partial strategyproofness, a new, relaxed notion of strategyproofness, to study the incentive properties of non-strategyproof assignment mechanisms. 
In this paper, we present results pertaining to local sufficiency for partial strategyproofness: 
We show that, for any $r \in [0,1]$, $r$-local partial strategyproofness implies $r^2$-partial strategyproofness, and we show that this is the tightest polynomial bound for which a guarantee can be proven. 
Our results unify the two prior local sufficiency results for strategyproofness \citep{Carroll2012WhenAreLocalIncentiveConstraintsSufficient} and lexicographic dominance-strategyproofness \citep{Cho2012AxiomaticLocalVsGlobalSP}. 
\end{abstract}
\noindent \textbf{Keywords:}
Assignment,
Matching,
Strategyproofness, 
Lexicographic Dominance, 
Partial Strategyproofness,
Local Sufficiency

\medskip
\noindent\textbf{JEL:} %
%
%
C79, 
%
%
D47 
%
%
%
}

\section{Introduction}
%
We study \emph{assignment mechanisms}, which are procedures that assign indivisible objects to agents, taking into account the agents' preferences over objects but without the use of monetary transfers. 
Incentives for truthtelling play an important role in the research on such mechanism. 
\emph{Stochastic dominance strategyproofness} (\emph{SD-strategyproofness}) requires that a mechanism makes truthful reporting a dominant strategy for all agents, independent of their preference intensities. 
A weaker incentive concept is \emph{lexicographic dominance-strategyproofness} (\emph{LD-strategyproofness}), which requires that a mechanism makes truthful reporting a dominant strategy for those agents who have lexicographic preferences over the objects.\footnote{Informally, agents are said to have \emph{lexicographic preferences} if they prefer any (arbitrarily small) increase in their chance to obtain a more preferred object to any (even large) increase in their chances to obtain some less preferred objects. 
For example, an agent with lexicographic preferences would prefer to receive its first choice with a probability of $1\%$ or its third choice with a probability of $99\%$ to receiving its second choice for sure.} 
In \citep{MennleSeuken2017PSP_WP}, we have introduced \emph{partial strategyproofness}, a new, intermediate incentive concept that parametrizes the spectrum of incentive concepts between SD- and LD-strategyproofness. 

Under all three incentive concepts, agents are not restricted in the kinds of misreports that they may submit. 
Alternatively, one may suppose that agents only submit \emph{local} misreports that arise from their truthful preference order by inverting the order of two adjacently ranked objects. 
Restricting attention to local misreports gives rise to the notions of \emph{local SD-strategyproofness}, \emph{local LD-strategyproofness}, and \emph{local partial strategyproofness}. 
Obviously, each local incentive requirement is implied by its global counterpart. 
However, the question arises whether the opposite also holds. 
If local incentive constraints imply their global counterpart, we speak of \emph{local sufficiency}. 
This intriguing property can be used to greatly reduce the complexity of incentive concepts (see \citep{Carroll2012WhenAreLocalIncentiveConstraintsSufficient}. 
From a computational perspective, local sufficiency reduces algorithmic complexity because it reduces the number of constraints under the automated mechanism design paradigm \citep{Sandholm2003AutomatedMechanismDesign}. 

\citet{Carroll2012WhenAreLocalIncentiveConstraintsSufficient} and \citet{Cho2012AxiomaticLocalVsGlobalSP} proved local sufficiency for SD- and LD-strategyproofness, respectively. 
Thus, local incentive constraints are sufficient for the two limit concepts of partial strategyproofness. 
In this paper, we contribute local sufficiency results for partial strategyproofness. 
First, we prove that $r$-local partial strategyproofness always implies $r^2$-partial strategyproofness (for any $r \in [0,1]$). 
Second, we show that $r^2$ is the tightest polynomial bound for which such an implication can be guaranteed. 
Our results provide a unified proof for the two prior local sufficiency results and illustrate an interesting connection between local and global incentive constraints on the spectrum of incentive concepts between the two limit concepts.

\section{Formal Model}
\label{SEC:MODEL}
We use the same model as in \citep{MennleSeuken2017PSP_WP}: 
A \emph{setting} $(N,M,q)$ consists of
a set of \emph{agents} $N$ ($n=\#N$),
a set of \emph{objects} $M$ ($m=\#M$),
and a vector $q = (q_1,\ldots,q_m)$ of \emph{capacities}
(i.e., there are $q_j$ units object $j$ available).
We assume $n \leq \sum_{j \in M} q_j$ (i.e., there are not more agents than the total number of units); 
otherwise we include a dummy object with capacity $n$.
Each agent $i\in N$ has a strict \emph{preference order} $P_i$ over objects, where $P_i : a \succ b$ indicates that agent $i$ prefers object $a$ to object $b$.
Let $\mathcal{P}$ be the set of all possible preference orders.
A \emph{preference profile} $P = (P_i)_{i \in N}\in \mathcal{P}^N$ is a collection of preference orders from all agents and $P_{-i} \in \mathcal{P}^{N\backslash\{i\}}$ is a collection of preference orders of all agents except $i$.
We extend agents' preferences to lotteries via von Neumann-Morgenstern utility functions:
A \emph{utility function} $u_i : M \rightarrow \mathds{R}^+$ is \emph{consistent} with a preference order $P_i$ if $u_i(a) > u_i(b)$ whenever $P_i : a \succ b$, denoted $u_i \sim P_i$. 
$U_{P_i} = \{u_i ~|~ u_i \sim P_i\}$ denotes the set of all utility functions consistent with $P_i$.

A \emph{(random) assignment} is represented by an $n \times m$-matrix $x = (x_{i,j})_{i \in N, j\in M}$, where no object is assigned beyond capacity (i.e., $\sum_{i \in N} x_{i,j} \leq q_j$ for all $j \in M$) and each agent receives some object with certainty (i.e., $\sum_{j \in M} x_{i,j} = 1$ for all $i \in N$ and $x_{i,j} \geq 0$ for all $i \in N, j \in M$).
The value $x_{i,j}$ is the probability that agent $i$ gets object $j$.
An assignment $x$ is \emph{deterministic} if $x_{i,j} \in \{0,1\}$ for all $i\in N, j \in M$.
The $i^{\text{th}}$ row $x_i = (x_{i,j})_{j\in M}$ of $x$ is called the \emph{assignment vector} of $i$ (short: \emph{$i$'s assignment}).
The Birkhoff-von Neumann Theorem and its extensions \citep{Budishetal2013DesignRandomAllocMechsTheoryAndApp} ensure that for any random assignment we can find a lottery over deterministic assignments that implements its marginal probabilities. 
Finally, let $X$ and $\Delta(X)$ denote the spaces of all deterministic and random assignments, respectively.

A \emph{(random assignment) mechanism} is a mapping $ \varphi : \mathcal{P}^N \rightarrow \Delta(X)$ that selects an assignment based on a preference profile.
$\varphi_i(P_i,P_{-i})$ denotes the assignment of agent $i$ when $i$ reports $P_i$ and the other agents report $P_{-i}$.
The mechanism $\varphi$ is \emph{deterministic} if it selects deterministic assignments (i.e., $\varphi : \mathcal{P}^N \rightarrow X$).
Note that we only consider \emph{ordinal} mechanisms, where the assignment only depends on the reported preference profiles but is independent of the underlying utility functions.
If agent $i$ with utility function $u_i$ reports $P_i$ and the other agents report $P_{-i}$, then agent $i$'s expected utility is 
\begin{equation}
	\mathds{E}_{\varphi_i(P_i,P_{-i})}[u_i]
	= \sum_{j \in M} u_i(j) \cdot \varphi_{i,j}(P_i,P_{-i}).
\end{equation}
\section{Notions of Strategyproofness}
%
First, we define the standard notion of SD-strategyproofness. 
\begin{definition}[SD-Strategyproofness]
\label{DEF:SD_SP}
	For a preference order $P_i$ and two assignment vectors $x_i,y_i$, we say that \emph{$x_i$ stochastically dominates $y_i$ at $P_i$} if, for all objects $j\in M$, we have
	\begin{equation}
		\sum_{j'\in M\text{ s.t. }P_i:j'\succ j} x_{j'} \geq \sum_{j'\in M\text{ s.t. }P_i:j'\succ j} y_{j'}. 
	\end{equation}
	
	A mechanism $\varphi$ is \emph{stochastic dominance strategyproof} (\emph{SD-strategyproof}) if, 
	for all agents $i \in N$,
	all preference profiles $(P_i,P_{-i}) \in \mathcal{P}^N$,
	and all misreports $P_i' \in \mathcal{P}$, 
	$\varphi_i(P_i,P_{-i})$ stochastically dominates $\varphi_i(P_i',P_{-i})$ at $P_i$.	
\end{definition}
SD-strategyproofness can be equivalently defined in terms of expected utilities \citep{Erdil2014SPStochasticAssignment}: 
A mechanism is SD-strategyproof if and only if truthful reporting maximizes any agent's expected utility, independent of its preference order, its particular utility function, or the reports from the other agents. 

A second, weaker notion of strategyproofness requires that agents have a dominant strategy to report truthfully when they have lexicographic preferences over objects (i.e., they prefer any arbitrarily small increase in their chances to receive a more preferred object to any even large increase in their chance to receive any less preferred object). 
\begin{definition}[LD-Strategyproofness]
\label{DEF:LD_SP}
	For preference order $P_i\in \mathcal{P}$ and assignment vectors $x_i,y_i$, we say that \emph{$x_i$ lexicographically dominates} \emph{$y_i$ at $P_i$} if 
	either $x_i=y_i$, 
	or $x_{i,a}>y_{i,a}$ for some $a \in M$ and $x_{i,j} = y_{i,j}$ for all $j \in U(a,P_i) =\{j\in M ~|~ P_i:j \succ a\}$.

	A mechanism $\varphi$ is \emph{LD-strategyproof} if, 
	for all agents $i \in N$,
	all preference profiles $(P_i,P_{-i}) \in \mathcal{P}^N$, 
	and all misreports $P_i' \in \mathcal{P}$, 
	$\varphi_i(P_i,P_{-i})$ lexicographically dominates $\varphi_i(P_i',P_{-i})$ at $P_i$. 
\end{definition}
Obviously, LD-strategyproofness is implied by SD-strategyproofness but the opposite is not true. 

The third incentive requirement that we define is partial strategyproofness. 
Intuitively, a mechanism is partially strategyproof if it makes truthful reporting a dominant strategy for all agents who have sufficiently different values for any two different objects. 
Formally, this corresponds to strategyproofness on a particular domain restriction.
\begin{definition}[URBI]
\label{DEF:URBI}
A utility function $u_i$ satisfies \emph{uniformly relatively bounded indifference with respect to bound $r \in [0,1]$ (\URBIr)} if, 
for all objects $a,b\in M$ with $u_i(a) > u_i(b)$, we have
\begin{equation}
	r \cdot \left( u_i(a) -\min_{j \in M} u_i(j) \right) \geq u_i(b) -\min_{j \in M} u_i(j).
\label{EQ:URBI_CONSTRAINT}
\end{equation}
%
\end{definition}
\begin{definition}[Partially Strategyproof]
\label{DEF:PSP}
Given a setting $(N,M,q)$ and a bound $r \in [0,1]$,
a mechanism $\varphi$ is \emph{$r$-partially strategyproof} (\emph{in the setting $(N,M,q)$})
if, 
for all agents $i \in N$,
all preference profiles $(P_i,P_{-i}) \in \mathcal{P}^N$,
all misreports $P_i' \in \mathcal{P}$,
and all utility function $u_i \in U_{P_i} \cap \URBIr$,
we have
\begin{equation}
	\mathds{E}_{\varphi_i(P_i,P_{-i})}[u_i] \geq \mathds{E}_{\varphi_i(P_i',P_{-i})}[u_i].
\end{equation}
$\varphi$ is \emph{partially strategyproof} if it is $r$-partially strategyproof for some positive bound $r >0$.
\end{definition}
We have introduced partial strategyproofness in \citep{MennleSeuken2017PSP_WP}, where we have also shown that it is a meaningful relaxation of strategyproofness for assignment mechanisms. 
In particular, the \emph{degree of strategyproofness} (i.e., the value $r$) parametrizes the space of incentive requirements between SD-strategyproofness ($r=1$) and LD-strategyproofness ($r\searrow 0$). 
Thus, partial strategyproofness yields a spectrum of incentive concepts with SD- and LD-strategyproofness as upper and lower limit concepts. 

\medskip
To simplify an incentive requirement, mechanism designers may choose to consider only \emph{local} misreports and require truthful reporting to be a dominant strategy in this restricted strategy space. 
For the assignment domain, a natural notion of \emph{locality} arises when agents are limited to inverting the order of just one pair of consecutively ranked objects: 
For any preference order $P_i\in \mathcal{P}$, the \emph{neighborhood of $P_i$}, denoted $N_{P_i}$, consists of all preference orders that differ from $P_i$ by a swap of two consecutively ranked objects.\footnote{For example, $P_i' : b \succ a \succ c$ is in the neighborhood of $P_i: a \succ b \succ c$, but $P_i'': c\succ a \succ b$ and $P_i''': c\succ b \succ a$ are not.}
\begin{definition}[Local SD- \& Local LD-Strategyproofness]
A mechanism $\varphi$ is \emph{locally strategyproof} if, 
for all agents $i\in N$, 
all preference profiles $(P_i,P_{-i}) \in \mathcal{P}^N$, 
all misreports $P_i' \in N_{P_i}$ from the neighborhood of $P_i$, 
and all utility functions $u_i  \in U_{P_i}$ that are consistent with $P_i$, 
we have
\begin{equation}
	\mathds{E}_{(P_i,P_{-i})}[u_i] - \mathds{E}_{(P_i',P_{-i})}[u_i] \geq 0.
\end{equation}
$\varphi$ is \emph{locally LD-strategyproof} if $\varphi_i(P_i,P_{-i})$ lexicographically dominates $\varphi_i(P_i',P_{-i})$ for all agents $i$, preference profiles $(P_i,P_{-i})$, and local misreports $P_i' \in N_{P_i}$.
\end{definition}
Analogously, we can define a local variant of partial strategyproofness. 
\begin{definition}[Local Partial Strategyproofness]
\label{DEF:LOC_PSP}
Given a setting $(N,M,q)$ and a bound $r \in (0,1)$, 
a mechanism $\varphi$ is \emph{$r$-locally partially strategyproof} if, 
for all agents $i\in N$, 
all preference profiles $(P_i,P_{-i}) \in \mathcal{P}^N$, 
all misreports $P_i' \in N_{P_i}$ from the neighborhood of $P_i$, 
and all utility functions $u_i \in U_{P_i}\cap \URBIr$ that are consistent with $P_i$ and satisfy \URBIr, 
we have
\begin{equation}
	\mathds{E}_{\varphi_i(P_i,P_{-i})}[u_i] - \mathds{E}_{\varphi_i(P_i',P_{-i})}[u_i] \geq 0.
\end{equation}
We say that $\varphi$ is \emph{locally partially strategyproof} if it is $r$-locally partially strategyproof for some non-trivial $r > 0$.
\end{definition}
\section{Local Sufficiency}
Local incentive requirements, such as local SD-strategyproofness, are obviously not more demanding than their global counterparts. 
However, the question arises whether these concepts are strict relaxations or whether local incentive requirements are in fact sufficient to imply the respective global requirements. 
\citet{Carroll2012WhenAreLocalIncentiveConstraintsSufficient} and \citet{Cho2012AxiomaticLocalVsGlobalSP} proved local sufficiency for SD- and LD-strategyproofness, respectively, and Facts \ref{FACT:SP_LOCAL_SUFFICIENCY} and \ref{FACT:LD_LOCAL_SUFFICIENCY} summarize their results. 
\begin{fact}[\citeauthor{Carroll2012WhenAreLocalIncentiveConstraintsSufficient}, \citeyear{Carroll2012WhenAreLocalIncentiveConstraintsSufficient}]
\label{FACT:SP_LOCAL_SUFFICIENCY} 
Local SD-strategyproofness is sufficient for SD-strategyproofness.
\end{fact}
\begin{fact}[\citeauthor{Cho2012AxiomaticLocalVsGlobalSP}, \citeyear{Cho2012AxiomaticLocalVsGlobalSP}]
\label{FACT:LD_LOCAL_SUFFICIENCY} 
Local LD-strategyproofness is sufficient for LD-strategyproofness.
\end{fact}
Since local incentive constraints are obviously necessary for SD- and LD-strategyproofness, local sufficiency implies equivalence. 

We are now ready to formulate our local sufficiency results for partial strategyproofness. 
First, we observe that Fact \ref{FACT:LD_LOCAL_SUFFICIENCY} (in combination with other insights about partial strategyproofness) immediately yields a weak form of local sufficiency.
\begin{corollary}
\label{COR:WEAK_LOCAL_SUFFICIENCY_PSP}
Given a setting $(N,M,q)$, if a mechanism $\varphi$ is $r$-locally partially strategyproof for some $r>0$, then it is $r'$-partially strategyproof for some $r'>0$.
\end{corollary}
Corollary \ref{COR:WEAK_LOCAL_SUFFICIENCY_PSP} follows from the observation that local partial strategyproofness implies local LD-strategyproofness (Theorem 4 in \citep{MennleSeuken2017PSP_WP}), 
which implies LD-strategyproofness (by Fact \ref{FACT:LD_LOCAL_SUFFICIENCY}), 
which in turn implies partial strategyproofness (again by Theorem 4 in \citep{MennleSeuken2017PSP_WP}). 
However, the local bound $r$ and the global bound $r'$ are not necessarily the same. 
Since $r'$-partial strategyproofness implies $r'$-local partial strategyproofness, we must have $r'\leq r$, but $r'$ may still be substantially smaller than $r$.
Our next result establishes a precise connection between $r$ and $r'$.
\begin{theorem}
\label{THM:LOCAL_SUFFICIENCY_PSP}
Given a setting $(N,M,q)$, if a mechanism $\varphi$ is $r$-locally partially strategyproof, then it is $r^2$-partially strategyproof.
\end{theorem}
%
%
We give the proof of Theorem \ref{THM:LOCAL_SUFFICIENCY_PSP} in Appendix \ref{APP:PROOF_SUFFICIENCY}.

Theorem \ref{THM:LOCAL_SUFFICIENCY_PSP} means that $r$-local partial strategyproofness is sufficient to guarantee $r'$-partial strategyproofness for any $r' \leq r^2$.
As a special case, we obtain that $1$-local partial strategyproofness implies $1$-partial strategyproofness, the local sufficiency result for strategyproofness (Fact \ref{FACT:SP_LOCAL_SUFFICIENCY}).
Furthermore, considering a sequence of bounds $(r_k)_{k\geq 1}$ that approaches $0$, we obtain the local sufficiency result for LD-strategyproofness in the limit (Fact \ref{FACT:LD_LOCAL_SUFFICIENCY}). 
Thus, Theorem \ref{THM:LOCAL_SUFFICIENCY_PSP} unifies both prior results. 

The question remains whether Theorem \ref{THM:LOCAL_SUFFICIENCY_PSP} is tight or whether the bound $r' \leq r^2$ can be improved in any way. 
First, note that it is straightforward to construct a counter-example to show that exact equality (i.e., $r' = r$, and therefore equivalence) is out of the question, unless $r \in \{0,1\}$.
In fact, as we show in the next Theorem \ref{THM:LOCAL_SUFFICIENCY_PSP:TIGHT}, the bound $r' = r^2$ is tight in the sense that `2' is the smallest exponent for which a universal guarantee can be given. 
%
\begin{theorem}
\label{THM:LOCAL_SUFFICIENCY_PSP:TIGHT}
Given a setting $(N,M,q)$ with $m \geq 4$ objects, for any $\varepsilon > 0$ there exists a bound $r \in (0,1)$ and a mechanism ${\varphi}$ such that 
\begin{enumerate}[1.]
	\setlength{\itemsep}{0pt}
	\item ${\varphi}$ is $r$-locally partially strategyproof, but
	\item ${\varphi}$ is \emph{not} $r^{2-\varepsilon}$-partially strategyproof.
\end{enumerate} 
\end{theorem}
Tightness by Theorem \ref{THM:LOCAL_SUFFICIENCY_PSP:TIGHT} means that $r' = r^2$ is the \emph{best polynomial bound} that allows a general statement about local sufficiency of the partial strategyproofness concept. 
We give the proof of Theorem \ref{THM:LOCAL_SUFFICIENCY_PSP:TIGHT} in Appendix \ref{APP:PROOF_TIGHT}.
%
\begin{remark}
\label{rem:details_r_r}
Observe that the value $r$ in the counter-examples in the proof of Theorem \ref{THM:LOCAL_SUFFICIENCY_PSP:TIGHT} may depend on $\varepsilon$.
We leave the exploration of the relationship between $r$ and $r'$ for \emph{fixed} $r$ to future research. 
\end{remark}



\appendix
\section*{Appendix}
\section{Proof of Theorem \ref{THM:LOCAL_SUFFICIENCY_PSP}}
\label{APP:PROOF_SUFFICIENCY}

\begin{proof}[Proof of Theorem \ref{THM:LOCAL_SUFFICIENCY_PSP}]
We must verify that an $r$-locally partially strategyproof mechanism $\varphi$ satisfies the conditions for $r^2$-partial strategyproofness, i.e., 
for all agents $i\in N$, 
all preference profiles $P = (P_i,P_{-i}) \in \mathcal{P}^N$, 
all misreports $P_i' \in \mathcal{P}$, 
and all utility functions $u_i  \in U_{P_i}$ with $u_i \in \text{URBI}(r^2)$, 
the inequality
\begin{equation}
	\mathds{E}_{\varphi_i(P_i,P_{-i})}[u_i] - \mathds{E}_{\varphi_i(P_i',P_{-i})}[u_i] \geq 0
	\label{eq:local_sufficiency:incentive_condition_for_verification}
\end{equation}
holds. 
Without loss of generality, we can assume that $\min_{j \in M} u_i(j) = 0$, since the manipulation incentives are exactly the same for an agent with utility function $\tilde{u}_i = u_i - \min_{j\in M} u_i(j)$. 

To simplify notation, we fix an arbitrary combination of agent $i$, preference profile $(P_i^{T},P_{-i})$, misreport $P_i^{F}$, and utility function $u_i \in U_{P_i^{T}} \cap \text{URBI}(r^2)$ to satisfy these preconditions. 
We drop the index $i$ on the preference orders, utility functions, and mechanism, and we omit the preferences of the other agents. 
With this simplification, inequality (\ref{eq:local_sufficiency:incentive_condition_for_verification}) becomes 
\begin{equation}
	\mathds{E}_{\varphi(P^{T})}[u] - \mathds{E}_{\varphi(P^{F})}[u] \geq 0
	\label{eq:local_sufficiency:incentive_condition_for_verification:simplified}
\end{equation}

Recall that $U_{P^{T}}$ denotes the set of utility functions that are consistent with $P^{T}$, i.e., 
\begin{equation}
	U_{P^{T}} = \left\{w : M \rightarrow \mathds{R}^+ ~|~ w \sim P^{T}\right\},
\end{equation}
and $U$ denotes the utility space, i.e., the union of all consistent utility functions 
\begin{equation}
	U = \bigcup_{ P \in \mathcal{P}} U_{ P}. 
\end{equation}
We say that a utility function $w : M \rightarrow \mathds{R}^+$ \emph{implies indifference} between two different objects $a,b \in M$ if $w(a) = w(b)$, and we denote by $W = \{w : M\rightarrow \mathds{R}^+\}$ the \emph{extended utility space}, i.e., the set of all possible utility functions, including those that imply indifference. 

Given the fixed preference order $P^{T}$ and consistent utility function $u \in U_{ P^{T}}$, let $v$ be a utility function that is consistent with the misreport $P^{F}$ and let
\begin{equation}
	\text{co}(u,v) = \left\{u_{\alpha} = (1-\alpha) u + \alpha v~|~\alpha \in [0,1]\right\}
\end{equation}
be the convex line segment in $W$ that connects $u$ and $v$. 
This line segment \emph{starts} in $U_{P^{T}}$, then (for increasing $\alpha$) traverses the extended utility space $W$ and eventually \emph{ends} at $v$ in $U_{P^{F}}$. 
$\text{co}(u,v)$ is said to \emph{pass} a preference order $P$ if, for some value $\alpha \in [0,1]$, we have that $u_{\alpha}$ is consistent with $P$, or equivalently, if $u_{\alpha} \in U_{P}$. 
By construction, $\text{co}(u,v)$ passes a sequence of preference orders $P^{T} = P^0, P^1, \ldots, P^{K-1},P^K =  P^{F}$ in this order, i.e., as $\alpha$ increases, $u_\alpha$ is first consistent with $P^0$, then with $P^1$, etc., until it is consistent with $P^K =  P^{F}$. 
Note that intermittently, it is possible that $u_{\alpha}$ is not consistent with any preference order as it may imply indifference. 
By linearity, we have that, for any two objects $a,b \in M$ with $u(a) > u(b)$ but $v(a) < v(b)$, there exists a \emph{unique} $\alpha \in (0,1)$ for which $u_{\alpha}$ implies indifference between $a$ and $b$, for any smaller $\alpha^- < \alpha$ we have $u_{\alpha^-}(a) > u_{\alpha^-}(b)$, and for any larger $\alpha^+ > \alpha$ we have $u_{\alpha^+}(a) < u_{\alpha^+}(b)$. 

We are now ready to formally define two requirements that we use in the proof:
\begin{itemize}
	\setlength{\itemsep}{0pt}
	\item We say that $\text{co}(u,v)$ makes \emph{no simultaneous transitions} if, for any three different objects $a,b,c \in M$, we have 
		\begin{equation}
			\text{co}(u,v) \cap \{w \in W ~ | ~ w(a) = w(b) = w(c)\} = \emptyset.
		\end{equation}
		In words, for no value of the parameter $\alpha$ does $u_{\alpha}$ imply indifference between all three objects $a,b,c$. 
		Intuitively, this means that two consecutive preference orders $P^k,P^{k+1}$ in the sequence $(P^0,\ldots,P^K)$ differ by exactly one swap of two consecutively ranked objects. 
	\item We say that $\text{co}(u,v)$ \emph{passes $P^{k}$ in \URBIr} if it passes $P^k$ and there exists some $\alpha^k \in [0,1]$ such that $u_{\alpha^k} \in U_{P^k} \cap \URBIr$. 
This means that the line segment contains at least one utility function that is consistent with $P^k$ and in addition satisfies \URBIr.
\end{itemize}
Suppose in the following that $u$ is consistent with $ P^{T}$ and satisfies $\text{URBI}(r^2)$, and that the mechanism $\varphi$ is $r$-locally partially strategyproof. 
\begin{claim}
\label{claim:r_squared_sufficiency}
There exists $v \in U_{ P^{F}} \cap \URBIr$ such that
\begin{enumerate}[i.]
	\setlength{\itemsep}{0pt}
	\item\label{claim:r_squared_sufficiency:i} $\text{co}(u,v)$ makes no simultaneous transitions,
	\item\label{claim:r_squared_sufficiency:ii} if $\text{co}(u,v)$ passes a preference order $ P$, then it passes $ P$ in \URBIr.
\end{enumerate}
\end{claim}
Using Claim \ref{claim:r_squared_sufficiency}, we can now show the inequality 
\begin{equation}
	\mathds{E}_{\varphi(P^{T})}[u] - \mathds{E}_{\varphi(P^{F})}[u] \geq 0.
\end{equation}
We will show this by writing the left side as a telescoping sum over local incentive constraints, where all but the first and the last term cancel out, such that it collapses to yield the inequality. 
This idea is inspired by the proof of local sufficiency for strategyproofness in \citep{Carroll2012WhenAreLocalIncentiveConstraintsSufficient}.

Consider the utility function $v$ constructed in Claim \ref{claim:r_squared_sufficiency} and the convex line segment $\text{co}(u,v)$. 
Let $\alpha^0 = 0$, $\alpha^K = 1$, and for each $k \in \{0,\ldots,K\}$ let $\alpha^k $ be the parameters for which $u_{\alpha^k} \in U_{P^k} \cap \URBIr$, which exist by statement \ref{claim:r_squared_sufficiency:ii} in Claim \ref{claim:r_squared_sufficiency}. 
For any $k \in \{0,\ldots,K-1\}$, the preference orders $P^k$ and $P^{k+1}$ are neighbors of each other, i.e., $P^k \in N_{P^{k+1}}$ and $P^{k+1} \in N_{P^{k}}$ (by statement \ref{claim:r_squared_sufficiency:i} in Claim \ref{claim:r_squared_sufficiency}). 
Thus, by $r$-local partial strategyproofness of $\varphi$, we obtain
\begin{equation}
	\mathds{E}_{\varphi(P^{k})}[u_{\alpha^{k}}] 
	- \mathds{E}_{\varphi(P^{k+1})}[u_{\alpha^{k}}] 
	\geq 0
\end{equation}
and 
\begin{equation}
	\mathds{E}_{\varphi(P^{k+1})}[u_{\alpha^{k+1}}] 
	- \mathds{E}_{\varphi(P^{k})}[u_{\alpha^{k+1}}] 
	\geq 0.
\end{equation}
Multiplying by $\alpha^k$ and $-\alpha^{k+1}$, respectively, and adding both inequalities yields
\begin{equation}
	\mathds{E}_{\varphi(P^{k+1})}[\alpha^k u_{\alpha^{k+1}} - \alpha^{k+1} u_{\alpha^k}]
	- 
	\mathds{E}_{\varphi(P^k)}[\alpha^k u_{\alpha^{k+1}} - \alpha^{k+1} u_{\alpha^k}]
	\geq 0.
	%
\end{equation}
Now, observe that $\alpha^k u_{\alpha^{k+1}} - \alpha^{k+1} u_{\alpha^k} = (\alpha^k - \alpha^{k+1}) \cdot u $ and, therefore, 
\begin{equation}
	\mathds{E}_{\varphi(P^{k+1})}[u]
	- \mathds{E}_{\varphi(P^{k})}[u] 
	\geq 0
\end{equation}
for all $k \in \{0,\ldots,K-1\}$. 
Summing over all $k$, we get 
\begin{equation}
	\mathds{E}_{\varphi(P^{T})}[u]
	- \mathds{E}_{\varphi(P^{F})}[u] 
	= 
	\sum_{k=0}^{K-1} 
	\mathds{E}_{\varphi(P^{k+1})}[u]
	- \mathds{E}_{\varphi(P^k)}[u] 
	\geq 0
	%
	%
\end{equation}

\medskip
We now proceed to prove Claim \ref{claim:r_squared_sufficiency}. 
\begin{proof}[Proof of Claim \ref{claim:r_squared_sufficiency}] 
The proof of existence of $v$ is constructive. 
For a preference order $P$, the \emph{rank} of an object $j$ under $ P$ is the position that $j$ holds in the ranking, i.e., 
\begin{equation}
	\text{rank}_{ P}(j) = \# \left\{ j'\in M~|~ P: j' \succ j \right\} + 1.
\end{equation}
Define $v:M\rightarrow \mathds{R}^+$ by setting
\begin{equation}
	v(j) = C^{m-\text{rank}_{ P^{F}}(j)}
\end{equation}
for any $j \in M$ and some $C > 1$ (so that $v \in U_{P^{F}}$). 
Furthermore, observe that $v \in \URBIr$ for sufficiently large $C$, since for any $a, b \in M$ with $P': a \succ b$, we have
\begin{equation}
	\frac{v(b)-\min_{j\in M} v(j)}{v(a) - \min_{j\in M} v(j)} = \frac{C^{\text{rank}_{ P^{F}}(b)} - 1}{C^{\text{rank}_{ P^{F}}(a)} - 1} = o\left(1/C\right).
\end{equation}
It remains to be shown that, for sufficiently large $C$, statements \ref{claim:r_squared_sufficiency:i} and \ref{claim:r_squared_sufficiency:ii} in Claim \ref{claim:r_squared_sufficiency} hold. 

To prove both statements, we use the concept of the \emph{canonical transitions}: 
For any two preference orders $P,P' \in \mathcal{P}$, a \emph{transition from $P'$ to $P$} is a finite sequence of preference orders $P^0,\ldots,P^K$ such that 
\begin{itemize}
	\setlength{\itemsep}{0pt}
	\item $P^0 = P'$ and $P^K = P$, 
	\item for all $k \in \{0,\ldots,K-1\}$, we have $P^k \in N_{P^{k+1}}$.
\end{itemize}
Intuitively, such a transition resembles a series of consecutive swaps that transform the preference order $P'$ into the preference order $P$. 
The \emph{canonical transition (from $P'$ to $P$)} is a particular transition that is inspired by the bubble-sort algorithm:  
Initially, we set $P^0=P'$. 
The preference orders $P^1,\ldots,P^K$ are constructed in \emph{phases}. 
In the first phase, we identify the highest ranking object under $P$ that is not ranked in the same position under $P'$, say $j$. 
Then, we construct the preference orders $P^1,P^2,\ldots$ by swapping $j$ with the respective next more preferred objects. 
When $j$ has reached the same position under $P^k$ as under $P$, the first phase ends. 
Likewise, at the beginning of the second phase, we identify the object that is ranked highest under $P$ of those objects that are not ranked in the same positions under $P^k$.  
New preference orders are constructed by swapping this object up to its final position under $P$. 
Subsequent phases are analogous. 
The construction ends when the most recently created preference order $P^K$ and $P$ coincides with $P$. 
 
In addition, we formalize \emph{transition times}: 
Suppose that, for two objects $a,b\in M$, we have $P^{T}: a\succ b$ but $P^{F}:b \succ a$, such that $u(a) > u(b)$ but $v(a) < v(b)$. 
Recall that in this case, there exists a unique parameter $\alpha $ for which $u_{\alpha}(a) = u_{\alpha}(b)$, for any smaller $\alpha^- < \alpha$ we have $u_{\alpha^-}(a) > u_{\alpha^-}(b)$, and for any larger $\alpha^+ > \alpha$ we have $u_{\alpha^+}(a) < u_{\alpha^+}(b)$. 
The line segment $\text{co}(u,v)$ pierces the hyperplane of indifference between $a$ and $b$ at the point $u_{\alpha}$, i.e., it transitions from preference orders that rank $a$ above $b$ to preference orders that rank $b$ above $a$. 
Formally, the \emph{transition time} $\alpha(a,b,1)$ is the parameter for which $u_{\alpha(a,b,1)}(a) = u_{\alpha(a,b,1)}(b)$. 
Extending this notation, we define $\underline{\alpha}(a,b,r)$ as the \emph{first} time when $u_{\alpha}$ violates the $\URBIr$ constraint for $a \succ b$, i.e., 
\begin{equation}
	\underline{\alpha}(a,b,r) = \inf \left\{\alpha \in [0,1]~\left|~\frac{u_{\alpha}(b) - \min_{j\in M} u_{\alpha}(j)}{u_{\alpha}(a) - \min_{j\in M} u_{\alpha}(j)}> r\right.\right\},
\end{equation}
and $\overline{\alpha}(b,a,r)$ as the \emph{last} time when $u_{\alpha}$ violates the $\URBIr$ constraint for $ b \succ a$, i.e., 
\begin{equation}
	\overline{\alpha}(a,b,r) = \sup \left\{\alpha \in [0,1]~\left|~\frac{u_{\alpha}(a) - \min_{j\in M} u_{\alpha}(j)}{u_{\alpha}(b) - \min_{j\in M} u_{\alpha}(j)}> r\right.\right\}.
\end{equation}
Obviously, 
\begin{equation}
	\underline{\alpha}(a,b,r) < \alpha(a,b,1) < \overline{\alpha}(b,a,r),
\end{equation}
i.e., as $\alpha$ increases, $u_{\alpha}$ violates $\URBIr$ for $a\succ b$ at some time, then subsequently it transitions from $a\succ b$ to $b\succ a$, and finally it no longer violates the $\URBIr$ constraint for $b \succ a$. 

We are now ready to formulate Claims \ref{claim:co_passes_canonical_transition}, \ref{claim:passing_in_urbi_r}, and \ref{claim:urbi_r_with_no_swap}, which are needed to establish statement \ref{claim:r_squared_sufficiency:i} (no simultaneous transitions) and statement \ref{claim:r_squared_sufficiency:ii} (passing all preference orders in $\URBIr$) in Claim \ref{claim:r_squared_sufficiency}, respectively, and the fact that the only relevant pairs of objects are those that are ranked differently under $ P^{T}$ and $ P^{F}$.

\begin{claim}
\label{claim:co_passes_canonical_transition} 
For sufficiently large $C$, $\text{co}(u,v)$ induces the canonical transition 
\begin{equation}
	P^0 =  P^{T}, P^1,\ldots,P^{K-1},P^K =  P^{F}.
\end{equation}
\end{claim}
\begin{claim}
\label{claim:passing_in_urbi_r} 
For sufficiently large $C$, if $\alpha(a,b,1) < \alpha(c,d,1)$, then 
\begin{equation}
	\overline{\alpha}(a,b,r) \leq \underline{\alpha}(c,d,r).
\end{equation}
\end{claim}
\begin{claim}
\label{claim:urbi_r_with_no_swap} 
If $P^{T}: a \succ b$ and $P^{F}: a \succ b$ and $u,v \in \URBIr$, then for all $\alpha \in [0,1]$
\begin{equation}
	\frac{u_{\alpha}(b)- \min_{j\in M} u_{\alpha}(j)}{u_{\alpha}(a)- \min_{j\in M} u_{\alpha}(j)} \leq r.
\end{equation}
\end{claim}

Since $\text{co}(u,v)$ induces a \emph{transition} by Claim \ref{claim:co_passes_canonical_transition}, we already know that for all pairs $(a,b) \neq(c,d)$ we have $\alpha(a,b,1) \neq \alpha(c,d,1)$. 
Thus, $\text{co}(u,v)$ makes no simultaneous transitions. 

If $a$ is preferred to $b$ under both $ P^{T}$ and $ P^{F}$, then, by Claim \ref{claim:urbi_r_with_no_swap}, the $\URBIr$ constraint for $a$ over $b$ is satisfied for any $\alpha$. 
Suppose now that 
$P^{T}: a \succ b$, 
$P^{T}: c \succ d$, 
$P^{F}: b \succ a$, 
$P^{F}: d \succ c$, 
and $\alpha(a,b,1) < \alpha(c,d,1)$. 
Then $\text{co}(u,v)$ enters a new set of consistent utility functions $U_{P^k}$ at time $\alpha(a,b,1)$, where $P^{k}$ differs from $P^{k-1}$ by a swap of $a$ and $b$, 
and $\text{co}(u,v)$ leaves $U_{P^k}$ at time $\alpha(c,d,1)$, where $P^k$ differs from $P^{k+1}$ by a swap of $c$ and $d$. 
In this case, the $\URBIr$ constraint for $b$ over $a$ is satisfied after time $\overline{\alpha}(a,b,r) > \alpha(a,b,1)$, and the $\URBIr$ constraint for $c$ over $d$ is satisfied before time $\underline{\alpha}(c,d,r) < \alpha(c,d,1)$. 
Claim \ref{claim:passing_in_urbi_r} yields that the constraint for $c$ over $d$ holds long enough for the constraint for $b$ over $a$ to be restored. 
Thus, at any time $\alpha^k \in [\overline{\alpha}(a,b,r),\underline{\alpha}(c,d,r)] \neq \emptyset$, both constraints are satisfied. 
Iterated application of this argument yields that, for any $k \in \{0,\ldots,K\}$, there exists some $\alpha^k$ for which $u_{\alpha^k}$ satisfies $\URBIr$ with respect to preference order $P^k$. 

This concludes the proof of Claim \ref{claim:r_squared_sufficiency}.
\end{proof}

\medskip
We now provide the proofs of Claims \ref{claim:co_passes_canonical_transition} and \ref{claim:passing_in_urbi_r}. 
Claim \ref{claim:urbi_r_with_no_swap} is obvious. 
\begin{proof}[Proof of Claim \ref{claim:co_passes_canonical_transition}]
First, we formulate an equivalent condition for $\text{co}(u,v)$ to induce the canonical condition in terms of transition times $\alpha(a,b,1)$. 
%
%
\begin{claim}
\label{claim:canonical_transition_equivalence} 
The following are equivalent:
\begin{enumerate}
	\setlength{\itemsep}{0pt}
	\item \label{claim:canonical_transition_equivalence:1} $\text{co}(u,v)$ induces the canonical transition 
		\begin{equation}
			P^0 =  P^{T}, P^1,\ldots,P^{K-1},P^K =  P^{F}.
		\end{equation}
	\item \label{claim:canonical_transition_equivalence:2} For any $a,b,c,d\in M$ with $P^{T}: a \succ b$, $P^{T}: c \succ d$, $P^{F}: b \succ a$, $P^{F}: d \succ c$,
		\begin{enumerate}[i.]
			\setlength{\itemsep}{0pt}
			\item \label{claim:canonical_transition_equivalence:2:i} if $P^{F}: b \succ d$, then $\alpha(a,b,1) < \alpha(c,d,1)$,
			\item \label{claim:canonical_transition_equivalence:2:ii} if $b = d$  and $P^{T}: c \succ a$, then $\alpha(a,b,1) < \alpha(c,d,1)$.
		\end{enumerate}
\end{enumerate}
\end{claim}
\begin{proof}[Proof of Claim \ref{claim:canonical_transition_equivalence}] 
First, we show sufficiency (``$\Rightarrow$''). 
To see that statement \ref{claim:canonical_transition_equivalence:2:i} holds, observe that, since $P^{F}: b\succ d$, $b$ will be brought up by bubble sort before $d$ is ever swapped up against another object. 
Since $P^{T}: c\succ d$, the swap of $c \leftrightarrow d$ is such a swap and, therefore, it has to occur after the swap $a \leftrightarrow b$.
Statement \ref{claim:canonical_transition_equivalence:2:ii} follows by observing that from $b=d$ and $P^{T}: c\succ a$ we get that $P^{T}: c \succ a \succ b$, but ultimately $P^{F}: b \succ (a,c)$. 
The bubble sort algorithm will bring $b$ up by swapping it with $a$ before it swaps $b$ and $c$. 

To see necessity (``$\Leftarrow$''), let $a \leftrightarrow b$ and $c \leftrightarrow d$ be two swaps that occur at $\alpha(a,b,1)$ and $\alpha(c,d,1)$, respectively. 
If $P^{F}: b \succ d$, then statement \ref{claim:canonical_transition_equivalence:2:i} implies that $a \leftrightarrow b$ occurs before $c \leftrightarrow d$, which is consistent with the canonical transition. 
By symmetry, the case $P^{F}:d \succ b$ also follows. 
Next, observe that any case not covered by this argument involves identity of $b$ and $d$, i.e., $b=d$. 
If $a = c$ as well, then there is nothing to show, so assume $P^{T}: a \succ c$, where \ref{claim:canonical_transition_equivalence:2:ii} implies the correct behavior. 
The last remaining case where $b = d$ and $P^{T}: c \succ a$ follows again by symmetry. 

This concludes the proof of Claim \ref{claim:canonical_transition_equivalence}.
\end{proof}
We now verify that the sequence of types through which $\text{co}(u,v)$ passes is indeed a canonical transition. 
Let $a,b,c,d \in M$ be such that $P^{T}: a \succ b$, $P^{T}: c \succ d$, $P^{F}: b \succ a$, $P^{F}: d\succ c$, and either $P^{F}: b \succ d$ (as in \ref{claim:canonical_transition_equivalence:2:i} of Claim \ref{claim:canonical_transition_equivalence}) or $b = d$ and $P^{T}: c\succ a$ (as in \ref{claim:canonical_transition_equivalence:2:ii} of Claim \ref{claim:canonical_transition_equivalence}).
We can write 
\begin{equation}
	\alpha(a,b,1) = \frac{u(a)-u(b)}{u(a)-u(b) + v(b)-v(a)} \text{~and~}\alpha(c,d,1) = \frac{u(c)-u(d)}{u(c)-u(d) + v(d)-v(c)},
\end{equation}
and we need to show that 
\begin{eqnarray}
	& & \alpha(a,b,1) < \alpha(c,d,1) \\
	& \Leftrightarrow & \left(u(a) - u(b)\right) \left( u(c)-u(d)+v(d)-v(c) \right) \\
	& & \hspace{1.6em} < \left(u(c) - u(d)\right) \left( u(a)-u(b)+v(b)-v(a) \right) \\
	& \Leftrightarrow & \left(u(a) - u(b)\right) \left( v(d)-v(c) \right) < \left(u(c) - u(d)\right) \left( v(b)-v(a) \right) \\
	&\Leftrightarrow & \frac{u(a)-u(b)}{u(c)-u(d)} < \frac{v(d)-v(c)}{v(b)-v(a)}. \label{eq:monotonicity_for_no_simul_transitions_last_condition}
\end{eqnarray}
If $P^{T}: b \succ d$, the left side of (\ref{eq:monotonicity_for_no_simul_transitions_last_condition}) grows faster than $C$, i.e., 
\begin{equation}
	\frac{v(d)-v(c)}{v(b)-v(a)} = \frac{C^{m - \text{rank}_{ P^{F}}(d)}-C^{m - \text{rank}_{ P^{F}}(c)}}{C^{m - \text{rank}_{ P^{F}}(b)}-C^{m - \text{rank}_{ P^{F}}(a)}} = \omega(C),
\end{equation}
since 
$\text{rank}_{ P^{F}}(b) < \text{rank}_{ P^{F}}(d)$, 
$\text{rank}_{ P^{F}}(b) < \text{rank}_{ P^{F}}(a)$, and
$\text{rank}_{ P^{F}}(d) < \text{rank}_{ P^{F}}(c)$. 
Similarly, if $b=d$ and $P^{T}: c \succ a$, we obtain that 
\begin{equation}
	\frac{v(d)-v(c)}{v(b)-v(a)} 
		= \frac{C^{m - \text{rank}_{ P^{F}}(b)}-C^{m - \text{rank}_{ P^{F}}(c)}}
			{C^{m - \text{rank}_{ P^{F}}(b)}-C^{m - \text{rank}_{ P^{F}}(a)}} 
		= \omega(C). 
\end{equation}
Since the right side in (\ref{eq:monotonicity_for_no_simul_transitions_last_condition}) is not small for sufficiently large $C$, we can ensure that $\alpha(a,b,1) < \alpha(c,d,1)$ whenever the statements \ref{claim:canonical_transition_equivalence:2:i} or \ref{claim:canonical_transition_equivalence:2:ii} in Claim \ref{claim:canonical_transition_equivalence} hold.
\end{proof}

\medskip
\begin{proof}[Proof of Claim \ref{claim:passing_in_urbi_r}]
First we define a conservative estimate for the violation times $\overline{\alpha}(a,b,r)$ and $\underline{\alpha}(c,d,r)$. 
Let 
\begin{equation}
	s(a,b,\alpha) = \frac{u_{\alpha}(b)}{u_{\alpha}(a)}
\end{equation}
and observe that $s(a,b,\alpha)$ is continuous and strictly monotonic in $\alpha$ and $s(a,b,\alpha(a,b,1)) = 1$. 
Thus, we can define the inverse $\alpha(a,b,s)$ for which $s(a,b,\alpha(a,b,s)) = s$ for any value of $s$ that is attained by $s(a,b,\alpha)$. 
In particular for $\alpha = 0$, $s(a,b,0) = \frac{u(b)}{u(a)} \leq r$ and for $\alpha = 1$, $s(a,b,1) = \frac{v(b)}{v(a)} > \frac{1}{r}$, so $\alpha(a,b,s)$ is well-defined for all values $s \in \left[r,\frac{1}{r}\right]$. 
In fact, we can solve
\begin{equation}
	\frac{u_{\alpha(a,b,s)}(b)}{u_{\alpha(a,b,s)}(a)} = s
\end{equation}
for $\alpha(a,b,s)$ and obtain the expression
\begin{equation}
	\alpha(a,b,s) = \frac{s u(a) - u(b)}{s u(a) - u(b) + v(b)-sv(a)}.
\end{equation}
Using $\min u_{\alpha} \geq 0$, 
\begin{equation}
		s(a,b,\alpha) = \frac{u_{\alpha}(b) }{u_{\alpha}(b) } \leq r 
\end{equation}
implies 
\begin{equation}
		\frac{u_{\alpha}(b) - \min_{j\in M} u_{\alpha}(j) }{u_{\alpha}(b) - \min_{j\in M} u_{\alpha}(j)} \leq r,
\end{equation}
and therefore, 
\begin{equation}
	 \overline{\alpha}(a,b,r) \leq \alpha(b,a,r)\text{ and }\alpha(c,d,r)\leq\underline{\alpha}(c,d,r).
\end{equation}
We now show that, for sufficiently large $C$, $\alpha(b,a,r) \leq \alpha(c,d,r)$ holds. 
Recall that we are considering objects $a,b,c,d \in M$, where $P^{T}: a \succ b$, $P^{T}: c\succ d$, $P^{F}: b \succ a$, and $P^{F}: d\succ c$, so that the required inequality can be rewritten equivalently as
\begin{equation}
	\alpha(b,a,r) \leq \alpha(c,d,r)  \Leftrightarrow \frac{u(a) - ru(b)}{ru(c) - u(d)} \leq \frac{rv(b)-v(a)}{v(d)-rv(c)}.
	\label{eq:equiv_condition_within_urbi_for_proof}
\end{equation}
By Claim \ref{claim:co_passes_canonical_transition}, $\text{co}(u,v)$ induces the canonical transition for sufficiently large $C$. 
Thus, by Claim \ref{claim:canonical_transition_equivalence}, $\alpha(a,b,1) < \alpha(c,d,1)$ holds if 
\begin{enumerate}[i.]
	\setlength{\itemsep}{0pt}
	\item \label{claim:passing_in_urbi_r:case:b_ultimately_pref_to_d} either $P^{F}: b \succ d$,
	\item \label{claim:passing_in_urbi_r:case:b_equal_d} or $b = d$ and $P^{T}: c \succ a$. 
\end{enumerate} 
In case \ref{claim:passing_in_urbi_r:case:b_ultimately_pref_to_d} we observe that the left side of (\ref{eq:equiv_condition_within_urbi_for_proof}) is constant, but the right side grows in $C$, i.e., it is in $\omega(C)$.
Therefore, (\ref{eq:equiv_condition_within_urbi_for_proof}) is ultimately satisfied for sufficiently large $C$. 

In case \ref{claim:passing_in_urbi_r:case:b_equal_d} the right side converges to $r$ (from below) as $C$ becomes large. 
Thus, it suffices to verify
\begin{eqnarray}
	& & \frac{u(a) - ru(b)}{ru(c) - u(b)} \leq r \\
	& \Leftrightarrow &  u(a) - ru(b) \leq r^2u(c) - ru(b) \\
	& \Leftrightarrow &  0 \leq r^2u(c) - u(a).
\end{eqnarray}
Using the assumption that $u$ satisfies $\text{URBI}(r^2)$, $\min_{j\in M} u(j) = 0$, and $P^{T}: c\succ a$, we get that 
\begin{equation}
	\frac{u(c)}{u(a)} \leq r^2~\Leftrightarrow~r^2u(c) - u(a) \geq 0. 
\end{equation}
This concludes the proof of Claim \ref{claim:passing_in_urbi_r}.
\end{proof}%
This concludes the proof of Theorem \ref{THM:LOCAL_SUFFICIENCY_PSP}.
%
%
\end{proof}

\section{Proof of Theorem \ref{THM:LOCAL_SUFFICIENCY_PSP:TIGHT}}
\label{APP:PROOF_TIGHT}

\begin{proof}[Proof of Theorem \ref{THM:LOCAL_SUFFICIENCY_PSP:TIGHT}]
Consider a mechanism $\varphi$ that selects the following assignments: 
\begin{eqnarray}
	\varphi(a\succ \ldots) & = & \left( \alpha, 0,0,1-\alpha \right), \\
	\varphi(b\succ \ldots) & = & \left( 0,\beta,0,1-\beta \right), \\
	\varphi(d\succ \ldots) & = & \left( 0,0,0,1 \right), \\
	\varphi(c\succ d \succ \ldots) & = & \left( 0,0,\gamma_c,1-\gamma_c \right), \\
	\varphi(c\succ a\succ d \succ b) & = & \left( 1-\gamma_c-\gamma_d,0,\gamma_c,\gamma_d \right), \\
	\varphi(c\succ b \succ \ldots) = \varphi(c\succ a\succ b \succ d) & = & \left( 1-\gamma_c-\gamma_d,\gamma_d,\gamma_c,0 \right)
\end{eqnarray}
for the objects $a,b,c,d$, respectively, where
\begin{eqnarray}
	\alpha,\beta,\gamma_c,\gamma_d & \in & [0,1], \\
	s & = & \frac{1}{r},\\
	\beta & = &  s \alpha, \\
	\gamma_c & = & \frac{(1-\alpha)}{\left(s-1\right)\left(s\left(s+1\right)-1\right)}, \\
	\gamma_d & = & \frac{s \left(s+1\right) (1-\alpha)}{s \left(s+1\right)-1}.
\end{eqnarray}
Observe that $\varphi$ is entirely specified by the values of $r$ and $\alpha$. 
We will now show that, for sufficiently small $r>0$, we can chose $\alpha$ such that 
\begin{enumerate}
	\setlength{\itemsep}{0pt}
	\item \label{item:construction:varphi:feasible} $\varphi$ is feasible,
	\item \label{item:construction:varphi:loc_psp} $\varphi$ is $r$-locally partially strategyproof,
	\item \label{item:construction:varphi:not_psp} but not $r^{2-\varepsilon}$-partially strategyproof. 
\end{enumerate}
First, we verify statement \ref{item:construction:varphi:feasible} that $\varphi$ is feasible. 
\begin{claim}
\label{claim:varphi_feasible}
For $s>1$, $\varphi$ is feasible if and only if $\alpha \in \left[\frac{s}{s^3-s+1},\frac{1}{s}\right]$.
\end{claim}
\begin{proof}[Proof of Claim \ref{claim:varphi_feasible}]
Note that for $s > 1$ and $\alpha < 1$, $\gamma_c$ and $\gamma_d$ are positive. 
We must ensure that $\beta = s \alpha \leq 1$, which is the case if and only if $\alpha \leq \frac{1}{s}$. 
Next, we give a condition for $\gamma_c+ \gamma_d \leq 1$, which in turn implies feasibility of the mechanism.
This inequality holds if and only if $\alpha \geq \frac{s}{s^3-s+1}$. 
Observing that $\frac{1}{s}>\frac{s}{s^3-s+1}$ for $s>1$, we have that the mechanism $\varphi$ is feasible if and only if $\alpha \in \left[\frac{s}{s^3-s+1},\frac{1}{s}\right] \neq \emptyset$.
\end{proof}

Second, we give equivalent conditions for $r$-local partial strategyproofness of $\varphi$, i.e., statement \ref{item:construction:varphi:loc_psp}. 
\begin{claim}
\label{claim:varphi_loc_psp}
For sufficiently small $r$, the following are equivalent:
\begin{enumerate}[i.]
	\setlength{\itemsep}{0pt}
	\item $\varphi$ is feasible and $r$-locally partially strategyproof,
	\item $\alpha \in I_s= \left[\frac{s^4 - s^3}{s^5+2s^4-s^2-s-1},\frac{s^3-s+\frac{s^2}{s-1}+1}{s^4+s^3-s^2+s+\frac{s^2}{s-1}}\right]$. 
\end{enumerate}
Furthermore, for sufficiently small $r>0$, $I_s \neq \emptyset$. 
\end{claim}
\begin{proof}[Proof of Claim \ref{claim:varphi_loc_psp}] 
We use Theorem 4 of \citep{MennleSeuken2017PSP_WP} to establish $r$-discounted dominance for any manipulation by just a swap, which in turn yields $r$-local partial strategyproofness. 
We only need to consider those swaps that lead to a change of the assignment, otherwise there is nothing to show. 
In the following, $\delta_k$ denotes the adjusted $k$th partial sum, i.e., for $P : j_1 \succ \ldots \succ j_m$, 
\begin{equation}
	\delta_k = \sum_{l=1}^k s^{k-l}\left(\varphi_{j_l}(P)-\varphi_{j_l}(P')\right) = r^{-k} \left(\sum_{l=1}^k r^{l}\left(\varphi_{j_l}(P)-\varphi_{j_l}(P')\right)\right).
\end{equation}
Observe that positivity of $\delta_1, \delta_2, \delta_3$ is equivalent to $r$-partial dominance of $\varphi_{j_l}(P)$ over $\varphi_{j_l}(P')$ at $P$ by Theorem 4 of \citep{MennleSeuken2017PSP_WP}. 
Table \ref{tbl:case_local_change_of_report} lists all the cases we need to consider. 
\begin{table}%
\begin{center}
\begin{tabular}{|l||c|c|c|c|c|c|c|c|c|}
	\hline
	\textbf{Preference report} & \textbf{I} & \textbf{II} & \textbf{III} & \textbf{IV} & \textbf{V} & \textbf{VI} & \textbf{VII} & \textbf{VIII} & \textbf{IX} \\
	\hline
	\hline
	$a\succ \ldots$ 																				& * & * & * & * &   &   &   &   &   \\
	\hline
	$b\succ \ldots$ 																				& * &   &   &   & * & * &   &   &   \\
	\hline
	$d\succ \ldots$ 																				&   & * &   &   & * &   & * &   &   \\
	\hline
	$c\succ d \succ \ldots$ 																&   &   &   &   &   &   & * & * &   \\
	\hline
	$c\succ a \succ d\succ b $															&   &   & * &   &   &   &   &   & * \\
	\hline
	$c \succ b\succ \ldots$ or $c \succ a \succ b \succ d$ 	&   &   &   & * &   & * &   & * & * \\
	\hline
\end{tabular}
\end{center}
\caption{Cases for local manipulations}
\label{tbl:case_local_change_of_report}
\end{table}
\begin{enumerate}[I.] 
	\setlength{\itemsep}{0pt}
	\item 
	\begin{itemize}
		\setlength{\itemsep}{0pt}
		\item $a \succ b \succ \ldots \rightsquigarrow b \succ a \succ \ldots$ : 
			\begin{eqnarray}
				\delta_1 & = &  \alpha \geq 0, \\
				\delta_2 & = & s \alpha - \beta = 0 \geq 0, \\
				\delta_3 & = & (1-\alpha) - (1-\beta) = \beta - \alpha \geq 0. 
			\end{eqnarray}
			For $\delta_3$, we assumed that the third choice is $d$, otherwise there is nothing to show. 
		\item $b \succ a \succ \ldots \rightsquigarrow a \succ b \succ \ldots$ : 
			\begin{eqnarray}
				\delta_1 & = &  \beta \geq 0, \\
				\delta_2 & = & s \beta - \alpha = \alpha(s^2-1) \geq 0, \\
				\delta_3 & = & s^2\beta - s \alpha + \alpha - \beta = \alpha (s^3 -2s +1) \geq 0. 
			\end{eqnarray}
			For $\delta_3$, we assumed that the third choice was $d$, otherwise there is nothing to show.
		\end{itemize}
	\item  $a \succ \ldots \leftrightsquigarrow d \succ \ldots$ : $\varphi(a\succ\ldots)$ first-order stochastically dominates $\varphi(d \succ \ldots)$ for all preference orders where $a$ is preferred to $d$, and vice versa. 
	\item \begin{itemize}
		\setlength{\itemsep}{0pt}
		\item $a \succ c \succ d \succ b \rightsquigarrow c \succ a \succ d \succ b$ : 
			\begin{eqnarray}
				\delta_1 & = &  \alpha - (1-\gamma_c - \gamma_d)  \\
					& = & \alpha - 1 + (1-\alpha)\left(\frac{(s-1)^{-1}+s(s+1)}{s(s+1)-1}\right) \geq 0,
			\end{eqnarray}
			since
			\begin{equation}
				(s-1)^{-1}+s(s+1)\geq s(s+1)-1~\Leftrightarrow~(s-1)^{-1}\geq -1. 
			\end{equation}
			\begin{eqnarray}
				\delta_2 & = & s(\alpha - 1 + \gamma_c + \gamma_d) - \gamma_c \\
					& = & (1-\alpha) s \left( \frac{(s-1)^{-1} +s(s+1) - (s-1)^{-1}s^{-1}}{s(s+1)-1}-1\right) \\
					& = & (1-\alpha) s \left( \frac{s(s+1)+s^{-1}}{s(s+1)-1}  - 1 \right) \geq 0, \\
				\delta_3 & = & s\delta_2 \geq 0. 
			\end{eqnarray}
		\item $c \succ a \succ d \succ b \rightsquigarrow a \succ c \succ d \succ b$ : 
			\begin{eqnarray}
				\delta_1 & = & \gamma_c \geq 0, \\
				\delta_2 & = & s \gamma_c +1 -\gamma_c - \gamma_d - \alpha \\
					& = & (a-\alpha) \left(1+ \frac{1-s(s+1)}{s(s+1)-1}\right) = 0 \\
				\delta_3 & = & \gamma_d-1+\alpha, \\
					& = & (1-\alpha)\left(\frac{s(s+1)}{s(s+1)-1} -1 \right) \geq 0.
			\end{eqnarray}		
		\end{itemize}
	\item \begin{itemize}
		\setlength{\itemsep}{0pt}
		\item $a \succ c \succ b \succ d \rightsquigarrow c \succ a \succ b \succ d$ : 
			\begin{eqnarray}
				\delta_1 & = &  \alpha - (1-\gamma_c - \gamma_d) \geq 0, \\
				\delta_2 & = & s(\alpha - 1 + \gamma_c + \gamma_d) - \gamma_c \geq 0, 
			\end{eqnarray}
			as in case III, and 
			\begin{eqnarray}
				\delta_3 & = & s^2(\alpha - 1 + \gamma_c + \gamma_d) - s \gamma_c + (1-\alpha) - \gamma_d \\
					& = & (1-\alpha)\left(1-s^2 + \frac{s + (s^2-1)s(s+1)}{s(s+1)-1}\right) \\
					& = & (1-\alpha)\left(\frac{s^2+s-1}{s(s+1) -1}\right) = 1-\alpha \geq 0. 
			\end{eqnarray}
		\item $c \succ a \succ b \succ d \rightsquigarrow a \succ c \succ b \succ d$ : 
			\begin{eqnarray}
				\delta_1 & = & \gamma_c \geq 0 \\
				\delta_2 & = & s \gamma_c + (1-\gamma_c-\gamma_d) - \alpha \\
					& = & (1-\alpha)\left( 1+\frac{\frac{s}{s-1} - \frac{1}{s-1} - s(s+1)}{s(s+1)-1} \right) \\ 
					& = & (1-\alpha)\left( 1+\frac{1 - s(s+1)}{s(s+1)-1} \right) = (1-\alpha)(1-1) = 0 \geq 0, \\ 
				\delta_3 & = & 0 + \gamma_d \geq 0.
			\end{eqnarray}
		\end{itemize}
	\item $b \succ \ldots \leftrightsquigarrow d \succ \ldots$ : $\varphi(b\succ\ldots)$ first-order stochastically dominates $\varphi(d \succ\ldots)$ for all preference orders where $b$ is preferred to $d$, and vice versa. 
	\item \begin{itemize}
		\setlength{\itemsep}{0pt}
		\item $b \succ c \succ \ldots \rightsquigarrow c \succ b \succ \ldots$ : 
			We begin with $\delta_3$ as its positivity also implies positivity of $\delta_1$ and $\delta_2$. 
			Furthermore, the strictest condition arises from the preference order $b\succ c \succ a \succ d$.
			\begin{eqnarray}
				\delta_3 & = & s^2(\beta - \gamma_d) + s (-\gamma_c) + (-1+\gamma_c+\gamma_d) \\
					& = & \alpha \left(\frac{s^5+2s^4-s^2-s-1}{s(s+1)-1}\right) - \frac{s^4 - s^3}{s(s+1)-1} \geq 0
			\end{eqnarray}
			holds if and only if 
			\begin{equation}
				\alpha \geq \frac{s^4 - s^3}{s^5+2s^4-s^2-s-1}.
			\end{equation}
		\item $c \succ b \succ \ldots \rightsquigarrow b \succ c \succ \ldots$ : 
			\begin{eqnarray}
				\delta_1 & = & \gamma_c \geq 0.
			\end{eqnarray}
			We can consider the case where $d$ is the third choice as this condition is stirctly stronger than if $a$ is the third choice. 
			It suffices to consider $\delta_3$ as its positivity implies positivity of $\delta_2$.
			\begin{eqnarray}
				\delta_3 & = & s^2\gamma_c + s \gamma_d - s \beta - (1-\beta) \\
					& = & \alpha\left(\frac{-s^4-s^3+s^2-s-\frac{s^2}{s-1}}{s(s+1)-1}\right) + \left(\frac{s^3-s+\frac{s^2}{s-1}+1}{s(s+1)-1}\right) \geq 0 
			\end{eqnarray}
			holds if and only if 
			\begin{equation}
				\alpha \leq \frac{s^3-s+\frac{s^2}{s-1}+1}{s^4+s^3-s^2+s+\frac{s^2}{s-1}}.
			\end{equation}
		\end{itemize}
	\item $d \succ c \succ \ldots \leftrightsquigarrow c \succ d \succ \ldots$ : 
			$\varphi(d\succ\ldots)$ first-order stochastically dominates $\varphi(c \succ d \succ \ldots)$ for all preference orders where $d$ is preferred to $c$, and vice versa. 
	\item \begin{itemize}
		\setlength{\itemsep}{0pt}
		\item $c \succ d \succ b \succ a \rightsquigarrow c \succ b \succ d \succ a$ : 
			\begin{eqnarray}
				\delta_1 & = & \gamma_c-\gamma_c \geq 0, \\
				\delta_2 & = & 1-\gamma_c - 0 \geq 0, \\
				\delta_3 & = & s(1-\gamma_c) + \gamma_d \geq 0.
			\end{eqnarray}
		\item $c \succ b \succ d \succ a \rightsquigarrow c \succ d \succ b \succ a$ : 
			\begin{eqnarray}
				\delta_1 & = & \gamma_c-\gamma_c \geq 0, \\
				\delta_2 & = & \gamma_d - 0 \geq 0, \\
				\delta_3 & = & s\gamma_d -(1-\gamma_c) \\
					& = & \alpha\left( \frac{-s^2(s+1)-\frac{1}{s-1}}{s(s+1)-1}\right) + \left( \frac{s^2(s+1) - \frac{1}{s-1} - s(s+1)+1}{s(s+1)-1}\right),
			\end{eqnarray}
			which is positive if and only if 
			\begin{equation}
				\alpha \leq \frac{s^3-s+1-\frac{1}{s-1}}{s^3+s^2-\frac{1}{s-1}}.
			\end{equation}
		\end{itemize}
	\item \begin{itemize}
		\setlength{\itemsep}{0pt}
		\item $c \succ a \succ d \succ b \rightsquigarrow c \succ a \succ b \succ d$ : 
			\begin{eqnarray}
				\delta_1 & = & \gamma_c-\gamma_c \geq 0, \\
				\delta_2 & = & 1-\gamma_c-\gamma_d - 1 +\gamma_c+\gamma_d \geq 0, \\
				\delta_3 & = & \gamma_d \geq 0.
			\end{eqnarray}
		\item $c \succ a \succ b \succ d \rightsquigarrow c \succ a \succ d \succ d$ : 
			\begin{eqnarray}
				\delta_1 & = & \gamma_c-\gamma_c \geq 0, \\
				\delta_2 & = & 1-\gamma_c-\gamma_d - 1 +\gamma_c+\gamma_d \geq 0, \\
				\delta_3 & = & \gamma_d \geq 0.
			\end{eqnarray}		
		\end{itemize}
\end{enumerate}
In summary, all local incentive constraints are satisfied if and only if 
\begin{equation}
	\frac{s^4 - s^3}{s^5+2s^4-s^2-s-1}
	\leq
	\alpha 
	\leq
	\min\left\{
		\frac{s^3-s+1-\frac{1}{s-1}}{s^3+s^2-\frac{1}{s-1}},
		\frac{s^3-s+\frac{s^2}{s-1}+1}{s^4+s^3-s^2+s+\frac{s^2}{s-1}}
	\right\}.
\end{equation}
The stronger upper bound is the second: 
Asymptotically, as $s$ grows, it behaves like $\frac{1}{s+1}$, which converges to 0, while the first bound converges to 1. 
The stronger upper bound is also stronger than the upper bound for feasibility, since $\frac{1}{s+1}$ is smaller than $\frac{1}{s}$
The lower bound behaves like $\frac{1}{s+2}$, which is greater than $\frac{1}{s^2-1}$, the asymptotic of the lower bound for feasibility. 
Finally, observe that the lower bound behaves like $\frac{1}{s+2}$, which is strictly less than the asymptotic of the upper bound $\frac{1}{s+1}$. 
Thus, for sufficiently large $s$, $\alpha$ can be chosen such that $\varphi$ is $r$-locally partially strategyproof, which in turn implies feasibility. 

This concludes the proof of Claim \ref{claim:varphi_loc_psp}.
\end{proof}

It remains to be shown that, for given $\varepsilon>0$, there exist $r$ and $\alpha$ such that $\varphi$ is $r$-locally partially strategyproof (and therefore feasible), but not $r^{2-\varepsilon}$-partially strategyproof, i.e., statement \ref{item:construction:varphi:not_psp}. 
To see this, we let $\tilde{s}=s^{2-\varepsilon}$ and consider the preference order $a\succ b \succ c \succ d$ and the non-local misreport $c\succ a \succ b \succ d$. 
If $\varphi$ is $\tilde{r}$-partially strategyproof, then in particular we must have $\delta_3 \geq 0$ for this manipulation. 
However, extensive algebraic transformations yield
\begin{eqnarray}
	\delta_3 & = & \tilde{s}^2 \left(\alpha - 1 + \gamma_c + \gamma_d\right) + \tilde{s}\left(-\gamma_d\right) + \left(-\gamma_c \right) \\
		& = & (1-\alpha) \left(\frac{-s^{5-\varepsilon} + s^{5-2\varepsilon} +s^{3-\varepsilon} -1}{s^3-2s+1}\right).
\end{eqnarray}
Since the leading term with exponent $5-\varepsilon$ has negative sign, this value is negative for sufficiently large $s$, and this negativity of $\delta_3$ is independent of $\alpha$. 

In conclusion, given a value of $\varepsilon > 0$, we can find $r>0$ and $\alpha \in (0,1)$ such that the resulting mechanism $\varphi$ is feasible and $r$-locally partially strategyproof, but it is not $r^{2-\varepsilon}$-partially strategyproof. 

This concludes the proof of Theorem \ref{THM:LOCAL_SUFFICIENCY_PSP:TIGHT}.
\end{proof}
\end{document}